\documentclass[letterpaper, 10 pt, conference]{ieeeconf}  % Comment this line out
\IEEEoverridecommandlockouts                              % This command is only
\usepackage{graphics} % for pdf, bitmapped graphics files
\usepackage{amsmath} % assumes amsmath package installed
\usepackage{amssymb}  % assumes amsmath package installed
\usepackage{tikz-cd}
\usepackage{cite}

\usepackage[ruled]{algorithm2e}

\newtheorem{theorem}{Theorem}
\newtheorem{definition}{Definition}
\newtheorem{example}{Example}

\newtheorem{lemma}{Lemma}

\newtheorem{assumption}{Assumption}

\newcommand{\join}{\vee}
\newcommand{\meet}{\wedge}
\newcommand{\bigjoin}{\bigvee}
\newcommand{\bigmeet}{\bigwedge}
\newcommand{\lattice}[1]{\mathbf{#1}}

\newcommand{\graph}[1]{#1}
\newcommand{\kripke}[1]{\mathcal{#1}}

\newcommand{\define}[1]{\textit{#1}}
\renewcommand{\vec}[1]{\mathbf{#1}}
\newcommand{\relation}[1]{\mathcal{#1}}
\newcommand{\bool}{2}
\renewcommand{\phi}{\varphi}
\newcommand{\sheaf}[1]{\mathcal{#1}}

\renewcommand{\geq}{\geqslant}
\renewcommand{\preceq}{\preccurlyeq}
\renewcommand{\succeq}{\succcurlyeq}
\newcommand{\N}{\{0,1,2,\dots \}}
\newcommand{\R}{\mathbb{R}}
\DeclareMathOperator{\id}{id}
\makeatletter
\newcommand{\fc}{%
  \mathrel{\mathpalette\fc@\relax}%
}

\newcommand{\fc@}[2]{%
  \sbox\z@{$#1\lhd$}%
  \sbox\tw@{$#1\leqslant$}%
  \dimen@=\ht\tw@
  \advance\dimen@-\ht\z@
  \ifx#1\displaystyle
    \advance\dimen@ .2pt
  \else
    \ifx#1\textstyle
      \advance\dimen@ .2pt
    \fi
  \fi
  \ooalign{\raisebox{\dimen@}{$\m@th#1\lhd$}\cr$\m@th#1\leqslant$\cr}%
}
\makeatother

\newcommand{\kripkeframe}{\mathbf{F}}

\title{\LARGE \bf Diffusion of Information on Networked Lattices by Gossip}
\author{Hans Riess\thanks{University of Pennsylvania; Department of Electrical/Systems Engineering; \texttt{hmr@seas.upenn.edu}} \and Robert Ghrist \thanks{University of Pennsylvania; Departments of Mathematics and Electrical/Systems Engineering; \texttt{ghrist@math.upenn.edu}}\thanks{\scriptsize This material is based upon work supported by the Under Secretary of Defense for Research and Engineering (Research Technology \& Laboratory Directorate/Basic Research Office) under Grant No. HQ00342110001. The views expressed do not necessarily reflect the official policies of the Department of Defense nor does mention of trade names, commercial practices, or organizations imply endorsement by the U.S. Government.}}% <-this % stops a space

\begin{document}

\maketitle
% \thispagestyle{empty}
% \pagestyle{empty}

%%%%%%%%%%%%%%%%%%%%%%%%%%%%%%%%%%%%%%%%%%%%%%%%%%%%%%%%%%%%%%%%%%%%%%%%%%%%%%%%
\begin{abstract}
We study time-dependent dynamics on a network of order lattices, where structure-preserving lattice maps are used to fuse lattice-valued data over vertices and edges. The principal contribution is a novel asynchronous Laplacian, generalizing the usual graph Laplacian, adapted to a network of heterogeneous lattices. The resulting gossip algorithm is shown to converge asymptotically to stable ``harmonic'' distributions of lattice data. This general theorem is applicable to several general problems, including lattice-valued consensus, Kripke semantics, and threat detection, all using asynchronous local update rules.
\end{abstract}
%%%%%%%%%%%%%%%%%%%%%%%%%%%%%%%%%%%%%%%%%%%%%%%%%%%%%%%%%%%%%%%%%%%%%%%%%%%%%%%%

% The advent of sensor, wireless ad hoc and peer-to-peer networks has necessitated the design of distributed and fault-tolerant computation and information exchange algorithms. This is mainly because such networks are constrained by the following operational characteristics: i) they may not have a centralized entity for facilitating computation, communication, and time-synchronization, ii) the network topology may not be completely known to the nodes of the network, iii) nodes may join or leave the network (even expire), so that the network topology itself may change, and iv) in the case of sensor networks, the computational power and energy resources may be very limited.

%%%%%%%%%%%%%%%%%%%%%%%%%%%%%%%%%%%%%%%%%%%%%%%%%%%%%%%%%%%%%%%%%%%%%%%%%%%%%%%%
\section{Introduction}
\label{sec:intro}
%%%%%%%%%%%%%%%%%%%%%%%%%%%%%%%%%%%%%%%%%%%%%%%%%%%%%%%%%%%%%%%%%%%%%%%%%%%%%%%%
The use of the graph Laplacian to diffuse information over networks is well-established in classical and contemporary work ranging from opinion dynamics \cite{taylor1968towards} to distributed multi-agent consensus \cite{degroot1974} and control \cite{preciado2010distributed}, synchronization \cite{preciado2015sync,sepulchre2005collective}, flocking \cite{tanner2007flocking}, and much more. In the past decade, Laplacians that are adapted to handle vector-valued data, such as graph connection Laplacians \cite{singer_vector_2012,BandeiraSS13} or matrix-weighted Laplacians \cite{tuna16synchronization}, have been revolutionary in signal processing processing \cite{ortega2018graph,shuman2013emerging} and machine learning \cite{welling2016semi,ruiz2021graph}.

While the ultimate form of a generalized Laplacian is as yet not present in applications, there are hints of a broader theory finding its way from algebraic topology to data science. The Laplacian from calculus class and the graph Laplacian are two extreme examples of a \define{Hodge Laplacian}. These are operators which act diffusively on data structures called \define{sheaves} \cite{bredon1997sheaf}: see \S\ref{sec:sheaves} for brief details.

The present work is motivated by extending recent work on distributed consensus and data fusion from the setting of vector-valued data to that of data valued in more general partially-ordered sets and, specifically, lattices (in the algebraic as opposed to discrete sense): see \S\ref{sec:lattices}. The fidelity with which lattices (including Boolean algebras) can model logical structures makes them appealing for representing distributed systems with complex logical behaviors, such as preference posets \cite{janis2015poset}, robust linear temporal logic \cite{anevlavis2022being} and discrete signal processing \cite{puschel2021discrete}. Maps between lattices preserving structure and their residuals model the interchange of information between neighbors: see \S\ref{sec:tarski}.

% ///////////////////////////////////////////////////////////////
\paragraph*{Related work}
The consensus literature is vast and includes Laplacian-based protocols \cite{olfati2007consensus} and (asynchronous) gossip-based protocols \cite{kempe2003gossip,boyd2006randomized}. Several works consider consensus on general functions \cite{cortes2008distributed}, including max/min consensus \cite{nejad2009max,tahbaz2006one}. In the area of distributed computing, centralized branch-and-bound style algorithms for reaching agreement on lattices have been discovered \cite{zheng2021lattice}. Sheaves have shown promise in multi-agent systems, particularly from the perspective of concurrency \cite{goguen1992sheaf}, routing \cite{ghrist2013topological,moy2020path}, opinion dynamics \cite{hansen2021opinion} and sensor fusion \cite{robinson2017sheaves}.

The Bayesian approach to modeling knowledge/belief propagation via graphical models \cite{koller2009probabilistic} is standard but fundamentally different than the approach here.  Modal logics, particularly temporal logics, have seen numerous applications in model-checking \cite{baier2008principles} and control systems \cite{kantaros2019temporal,rodionova2021stl,riess2021temporal,kantaros2020reactive}. There are several use-cases of multimodal logics in the analysis of message-passing systems \cite{fagin2004reasoning}.

The authors previously defined a synchronous [Tarski] Laplacian and proved a Hodge-style fixed-point convergence result \cite{ghrist2020cellular}, which is extended here to the asynchronous setting. The \define{gossip algorithm} introduced in this paper generalizes an algorithm called the alternating algorithm \cite[\S 3]{cuninghame2003equation} introduced to synchronize event times of a pair of coupled discrete event systems, each described by a max-plus linear system \cite{cuninghame1991minimax}; our algorithm goes far beyond this in computing sections of (nearly) arbitrary lattice-valued network sheaves.

% We present a model of information aggregation in which agents’ informa-
% tion is represented through partitions over states of the world. We discuss three axioms,
% meet separability, upper unanimity, and non-imposition, and show that these three axi-
% oms characterize the class of oligarchic rules, which combine all of the information
% held by a pre-specified set of individuals.

% ///////////////////////////////////////////////////////////////
\paragraph*{Outline}
Background material (\S\ref{sec:bg}) and problem specifications (\S\ref{sec:problem}) are followed by details of a novel Laplacian for asynchronous communication on networks of lattices (\S\ref{sec:async}). It is here that the main results on stability and convergence are proved. The subsequent section (\S\ref{sec:semantics}) detail applications to multimodal logics, by using Kripke semantics, leading to a dual pair of {\em semantic} and {\em syntactic} Laplacians for diffusing knowledge and beliefs. This work ends (\S\ref{sec:examples}) with some simulation.
%%%%%%%%%%%%%%%%%%%%%%%%%%%%%%%%%%%%%%%%%%%%%%%%%%%%%%%%%%%%%%%%%%%%%%%%%%%%%%%%
\section{Background}
\label{sec:bg}
%%%%%%%%%%%%%%%%%%%%%%%%%%%%%%%%%%%%%%%%%%%%%%%%%%%%%%%%%%%%%%%%%%%%%%%%%%%%%%%%

% ------------------------------------------------------------------------------
\subsection{Lattices}
\label{sec:lattices}
% ------------------------------------------------------------------------------

Ordered sets model data types such as relations, concepts, rankings, matchings, concurrent events, as well as other taxonomies of information that are hierarchial, epistemic, or logical in nature: both in \S\ref{sec:semantics}. Lattices double as ordered sets and algebraic structures consisting of two ``merging'' operations called \define{meet} and \define{join}.

\begin{definition}
A \define{lattice} is a tuple $\lattice{Q} = (Q, \meet, \join, 0, 1)$. $Q$ is a set with binary operations, $\meet$ (meet) and $\join$ (join), satisfying the following axioms: 1) $\meet$/$\join$ are commutative and associative; 2) for all $x \in Q$, $x \meet x = x$ and $x \join x = x$; 3) there exist elements $\bot, \top \in Q$ such that $\bot$ and $\top$ are the identity of $\join$ and $\meet$ respectively; 4) for all $x, y \in Q$, $x = x \join (x \meet y) = x \meet (x \join y)$.
\end{definition}

%\begin{remark}
  Equivalently, lattices can be viewed as (partially) ordered sets $(\lattice{Q}, \preceq)$ with $x \preceq y \Leftrightarrow x \meet y = x$ or, equivalently, $x \succeq y \Leftrightarrow x \join y = x$. It will be useful to think of lattices as both partially-ordered sets and algebraic structures: the $\preceq$ notation will be crucial in proofs of all main results.
%\end{remark}

\begin{example}
    Suppose $S$ is a set. The powerset $2^S$ is a (Boolean) lattice $(2^S, \cap, \cup, \emptyset, S)$. The truth values $\vec{2} = \left(\{0,1\}, \meet, \join, 0, 1\right)$ is a lattice. Other important lattices are embedded in $2^S$ such as lattices representing ontologies \cite{wille1982restructuring}, partitions \cite{davey2002introduction}, rankings, preferences \cite{curello2019preference,janis2015poset}, and information-theoretic content \cite{shannon1953lattice}. The extended real line $\bar{\R} = \R \cup \{- \infty, \infty \}$ is a lattice with min and max. Cartesian products of lattices are lattices with the component-wise meet and join operations; the lattice $\bar{\R}^n$ and its matrix algebra is integral to the study of discrete event systems \cite{cuninghame1991minimax}, $\vec{2}^n$ to logic gates \cite{gilbert1954lattice}.
\end{example}

In order to work with systems of lattices, we will exploit \define{lattice maps} $\varphi:\lattice{P}\to\lattice{Q}$. Such a map is \define{order preserving} if $x\preceq y \Rightarrow \varphi(x)\preceq\varphi(y)$ and is \define{join preserving} if $\varphi(x\vee_{\lattice{P}} y)=\varphi(x)\vee_{\lattice{Q}}\varphi(y)$ and $\varphi(\bot_\lattice{P})=\bot_\lattice{Q}$, omiting subscripts when context dictates. Join preserving maps are automatically order preserving. A dual definition of \define{meet preserving} maps holds with similar consequences. 

% ------------------------------------------------------------------------------
\subsection{Network Sheaves}
\label{sec:sheaves}
% ------------------------------------------------------------------------------
Suppose $\graph{G} = (\graph{V}, \graph{E})$ is an undirected graph (possibly with loops). An edge between (not necessarily distinct) nodes $i$ and $j$ in $\graph{V}$ is denoted by an unordered concatenated pair of indices $ij=ji \in \graph{E}$. The set $\graph{N}_i = \{j: ij \in \graph{E} \}$ are the \define{neighbors} of $i$.

Given such a fixed network $\graph{G}$, we define a data structure over $\graph{G}$ taking values in lattices. Such a structure is an example of a \define{cellular sheaf} \cite{shepard1985cellular,curry2014sheaves} (though the details of sheaf theory are not needed here).
\begin{definition}
     A  \define{finite lattice-valued network sheaf} $\sheaf{F}$ over $\graph{G}$ is a data structure that assigns:
    \begin{enumerate}
        \item A finite lattice $(\sheaf{F}(i), \meet_i, \join_j)$ to every node $i \in \graph{V}$.
        \item A finite lattice $(\sheaf{F}(ij), \meet_{ij}, \join_{ij})$ to every edge $ij \in \graph{E}$.
        \item Join-preserving maps
        \begin{equation}
            \begin{tikzcd}
                \sheaf{F}(i) \arrow[r,"\sheaf{F}_{i \fc  ij}"] & \sheaf{F}(ij) & \sheaf{F}(j) \arrow[l,"\sheaf{F}_{j \fc  ij}", swap]
            \end{tikzcd}
        \end{equation}
        for every $ij \in \graph{E}$.
    \end{enumerate}
\end{definition}

In most applications we imagine, the maps from node data to edge data will be join-preserving. Thinking of a sheaf as a distributed system, the state of the system is given by an \define{assignment} of vertex data: a tuple $\mathbf{x} \in \prod_{i \in \graph{V}} \sheaf{F}(i)$ of choices of data $x_i\in\sheaf{F}(i)$ for each $i$. The data over the edges and the maps $\sheaf{F}_{i \fc  ij}$ are used to compare the compatibility of vertex data in an assignment. The following definition is crucial.
\begin{definition}
The \define{sections} of $\sheaf{F}$ are assignments $\mathbf{x}$ that are compatible: for every $ij \in \graph{E}$,
\begin{equation}
    \sheaf{F}_{i \fc  ij}(x_i) = \sheaf{F}_{j \fc  ij}(x_j).
\end{equation}
\end{definition}
The set of sections of $\sheaf{F}$ is denoted $\Gamma(\sheaf{F})$. These are the ``globally compatible'' states. In the simplest example of a \define{constant sheaf} (which assigns a fixed lattice to each vertex and edge, with identity maps between vertex and edge data), sections are precisely assignments of an identical element to each vertex (and edge): consensus over the network. 

\begin{example}
The true utility of a sheaf lies in heterogeneity. For example, assign to each vertex $i\in V$ a finite set $S_i$ and to each edge $ij\in E$ a finite set $S_{ij}$ and set-maps $S_i\rightarrow S_{ij} \leftarrow S_j$. This induces several interesting sheaves of lattices. The \define{powerset sheaf} assigns the powersets (lattices of all subsets with union and intersection) to vertices and edges, with induced maps. The \define{partition sheaf} assigns the partition lattices (partitions of set elements with partition union and refinement) to vertex and edge sets, again with induced maps. More examples from formal concept analysis are less well-known \cite{wille1982restructuring}, but very general.     
\end{example}

%%%%%%%%%%%%%%%%%%%%%%%%%%%%%%%%%%%%%%%%%%%%%%%%%%%%%%%%%%%%%%%%%%%%%%%%%%%%%
\section{Problem Formulation}
\label{sec:problem}
%%%%%%%%%%%%%%%%%%%%%%%%%%%%%%%%%%%%%%%%%%%%%%%%%%%%%%%%%%%%%%%%%%%%%%%%%%%%%

An abstract formulation is sufficient, but, for concreteness, consider a scenario in which a collection of geographically dispersed agents collect, process, and communicate data based on local sensing. Proximity gives rise to a communications network, modeled as an undirected graph $\graph{G}$. The data are lattice-valued, but each agent $i\in V$ works within its lattice $\sheaf{F}(i)$, assumed to be finite. In order for two proximate agents $i$ and $j$ to communicate their individualized data (residing in $\sheaf{F}(i)$ and $\sheaf{F}(j)$ respectively), they must ``fuse'' their observiations into a common lattice $\sheaf{F}(ij)$ by means of structure-preserving lattice maps. Together, the system forms a network sheaf.

The problem envisioned is distributed consensus by means of asynchronous communication and updates. By consensus, we do not mean that everyone agrees on a particular fixed lattice value; rather, each node agrees upon choices of {\em local} data that, when translated and compared to all neighbors' data over communication edges, agree. In the context of a sheaf of lattices, this is precisely the condition of an assignment being a {\em section}. The synchronous version of this problem -- everyone communicates simultaneously with neighbors and updates immediately -- is solvable via the Tarski Laplacian as per \cite{ghrist2020cellular}. The asynchronous problem is our focus here. Select sensors broadcast their data to neighboring nodes according to some \define{firing sequence}. 

Denote by $\tau: \N \to 2^\graph{V}$ the firing sequence of selected nodes which broadcast as a function of (ordered, discrete) time. At $t \in \N$, active nodes $i\in \tau_t$ broadcast to each agent $j \in \graph{N}_{i}$.\footnote{If desired, one may choose a selection of egdes incident to the node and broadcast only to those neighbors. This results in more bookkeeping and a refined version of liveness, below, but does not substantially change the results or proofs.} No other nodes broadcast. 

This notion of a firing sequence suffices to cover asynchronous updating. The regularization of time to $\N$ is a convenience and does not impact results. The firing sequence is, in practice, not known {\em a priori}. This is of no consequence since our results will hold independent of the choice of firing sequence. The following assumptions will hold throughout.

\begin{assumption}[Latency-Free]\label{ass:nodelay}
    Within a single time instance $t$, nodes may broadcast (if firing), receive data (always), and compute (always). 
\end{assumption}
\begin{assumption}[Liveness]\label{ass:liveness}
    For all $i \in \graph{V}$ and for every $t \in \N$ there is a $t' \geq t$ such that $i \in \tau_{t'}$. As such, agents neither die nor are removed from the system.
\end{assumption}

\begin{assumption}[Cross-Talk]\label{ass:crosstalk}
    Suppose $i$ is an agent and $j, j' \in \graph{N}_i \cap \tau_t$ are active neighbors. Then, $j$ and $j'$ can simultaneously broadcast to $i$ without resulting in a fault.
\end{assumption}

%%%%%%%%%%%%%%%%%%%%%%%%%%%%%%%%%%%%%%%%%%%%%%%%%%%%%%%%%%%%%%%%%%%%%%%%%%%%%%%%%
%\subsection{Data fusion}
%%%%%%%%%%%%%%%%%%%%%%%%%%%%%%%%%%%%%%%%%%%%%%%%%%%%%%%%%%%%%%%%%%%%%%%%%%%%%%%%%

Under these assumptions, we wish to solve the following distributed asynchronous constrained agreement problem. Assume 1) a network $\graph{G}=(\graph{V},\graph{E})$; 2) a sheaf $\sheaf{F}$ of lattices over $\graph{G}$; 3) an firing sequence $\tau:\N\to 2^{\graph{V}}$ of broadcasts; and 4) an initial condition $\vec{x}[0]$, being an assignment of an element $x_v\in\sheaf{F}(v)$ to each agent $v\in\graph{V}$. Using only local communication subordinate to the firing sequence $\tau$, evolve the initial condition $\vec{x}[0]$ to a section $\vec{x}\in\Gamma(\sheaf{F})$.  

This problem has elements of consensus (because of the local agreement implied in a section) as well as data fusion (due to the lattice maps merging data from vertex lattices to edge lattices).

%%%%%%%%%%%%%%%%%%%%%%%%%%%%%%%%%%%%%%%%%%%%%%%%%%%%%%%%%%%%%%%%%%%%%%%%%%%%%%%%
\section{An Asynchronous Laplacian}
\label{sec:async}
%%%%%%%%%%%%%%%%%%%%%%%%%%%%%%%%%%%%%%%%%%%%%%%%%%%%%%%%%%%%%%%%%%%%%%%%%%%%%%%%
Our method for solving this asynchronous constrained agreement problem is to define an asynchronous harmonic flow on the sheaf by localizing the Tarski Laplacian of \cite{ghrist2020cellular}.

% ------------------------------------------------------------------------------
\subsection{The Tarski Laplacian}
\label{sec:tarski}
% ------------------------------------------------------------------------------
\noindent 
Throughout, $\sheaf{F}$ is a finite lattice-valued sheaf over a network $\graph{G}$. Our first step towards a Laplacian involves preliminaries on {\em residuals} \cite{blyth2014residuation}, also known as {\em Galois connections} \cite{ore1944galois}. These are a type of {\em adjoint} or lattice-theoretic analogue of the familiar {\em Moore-Penrose pseudoinverse} in matrix algebra.

\begin{definition}
    Given a join-preserving lattice map $\varphi:\lattice{P}\to\lattice{Q}$, its \define{residual} is the map $\varphi^{+}\colon\lattice{Q}\to\lattice{P}$ given by
    \[
        \varphi^{+}(p)=\bigjoin \{q~\vert~\varphi(q) \preceq p\}.
    \]
\end{definition}

Like an adjoint, it reverses the direction of the map, resembling a pseudoinverse more closely in some cases. The following two lemmas have straightforward proofs via definitions.

\begin{lemma}
    Suppose $\varphi:\lattice{P}\to\lattice{Q}$ is join-preserving and injective. Then $\varphi^{+}\circ\varphi = id$.
\end{lemma}

\begin{lemma}
\label{lem:residuated}
    For $\varphi:\lattice{P}\to\lattice{Q}$ join-preserving, the following identities hold:
    \begin{enumerate}
        \item for all $p \in \lattice{P}$, $\varphi^{+}\circ\varphi(p) \succeq p$; and
        \item  for all $p \in \lattice{P}$ and $q \in \lattice{Q}$,
        \begin{equation}
        \label{eq:swap}
            \varphi(p) \preceq q \, \Leftrightarrow \, p\preceq \varphi^{+}(q) .
        \end{equation}
    \end{enumerate}
\end{lemma}

A lattice-theoretic analogue of the graph Laplacian -- the \define{Tarski Laplacian} -- was introduced in \cite{ghrist2020cellular}. For $\sheaf{F}$ a network sheaf of lattices over $\graph{G}$ and $\vec{x}$ an assignment of vertex data, the Tarski Laplacian $L$ acts as:
\begin{equation}
\label{eq:Tarski}
  \left(L\vec{x}\right)_i 
  = 
  \bigmeet_{j \in \graph{N}_{i}} \sheaf{F}^{+}_{i\fc ij} \sheaf{F}_{j \fc  ij}(x_j).
\end{equation}

\noindent 
The key construct is to localize this operator subordinate to a firing sequence $\tau$.

\begin{definition}
\label{def:async-Tarski}
    The \define{asynchronous Tarski Laplacian} is the map 
    \[L: \N \times \prod_{i \in \graph{V}} \sheaf{F}(i) \to \prod_{i \in \graph{V}} \sheaf{F}(i) \]
    which acts on an assignment $\vec{x}$ as
    \begin{equation} 
    \label{eq:async-laplacian}
        \left(L_t\vec{x}\right)_i = \bigmeet_{j \in \graph{N}_{i} \cap \tau_t} \sheaf{F}^{+}_{i \fc ij}\sheaf{F}_{j \fc  ij}(x_j).
    \end{equation}
\end{definition}

This is a restriction of the Tarski Laplacian (\ref{eq:Tarski}) in that at time $t$, only immediate neighbors to broadcasting nodes are updated; all other nodes are unchanged. In the extreme of a firing sequence where all nodes broadcast at all times, the full Tarski Laplacian ensues.   

% ------------------------------------------------------------------------------
\subsection{Gossip and harmonic states}
\label{sec:heat}
% ------------------------------------------------------------------------------
\noindent 
The rationale for the {\em Laplacian} moniker lies in the efficacy of $L$ as a diffusion operator on states. Given an initial state $\vec{x}[0]$, \define{heat flow} is defined by the following discrete-time dynamical system:
\begin{align}
\label{eq:heatflow}
        \vec{x}[t+1] & = \left(\id \meet L_t \right)\vec{x}[t]
\end{align}
with initial condition $\vec{x}[0] \in \prod_{i \in \graph{V}} \sheaf{F}(i)$.

Heat flow is the analogue of iterating the \define{random-walk} or \define{Perron} operator  $I- \epsilon L$ (where $I$ is the identity matrix, $L$ the graph Laplacian matrix, $\epsilon > 0$) on scalar-valued data on a graph\cite{olfati2007consensus}. The principal result of this work is a type of ``Hodge Theorem:'' the \define{harmonic} states (the equilibria of $\id\meet L$ is a proxy for the kernel of the Laplacian) are exactly the globally consistent solutions to the sheaf. 

\begin{theorem}[Main Theorem]\label{thm:main}
For any finite lattice-valued network sheaf $\sheaf{F}$ with join-preserving structure maps and any firing sequence $\tau$ satisfying liveness, the sections of $\sheaf{F}$, $\Gamma(\sheaf{F})$, are precisely the time-independent solutions to heat flow (\ref{eq:heatflow}). 
\end{theorem}
\begin{proof}
    Suppose first that $\vec{x}[t]\in\Gamma(\sheaf{F})$ is a section. Then, for all $ij \in \graph{E}$,
    \[
      \sheaf{F}_{i \fc  ij}(x_i[t]) = \sheaf{F}_{j \fc  ij}(x_j[t]).
    \]
    Hence, by Lemma \ref{lem:residuated},

    \begin{align*}
    \left( L_t \vec{x}[t]\right)_i & = \bigmeet_{j \in \graph{N}_{i} \cap \tau_t} \sheaf{F}^{+}_{i\fc j}\sheaf{F}_{j \fc  ij}(x_j[t]) \\
    & = \bigmeet_{j \in \graph{N}_{i} \cap \tau_t} \sheaf{F}^{+}_{i\fc j} \sheaf{F}_{i \fc  ij}(x_i[t]) \\
    & \succeq x_i[t] .
    \end{align*}

    Then, $x_i[t+1] = \left(L_t\vec{x}[t]\right)_i \meet x_i[t] = x_i[t]$. Hence, $x_i[t+1]=x_i[t]$ for all $i \in \graph{V}$. Inducting in $t$, sections are time-independent.

    Conversely, suppose $\vec{x}[t]=\vec{x}[0]$ is a time-independent solution. Then, by (\ref{eq:heatflow}),
    \begin{equation}
        \left(L_t\vec{x}[0]\right)_i \succeq x_i[0]
    \end{equation}
    for all $i \in \graph{V}$ and all $t\in\N$. 
    Thus, for all $i \in \graph{V}$ and $j \in \graph{N}_{i} \cap \tau_t$,
    \begin{align*}
        \sheaf{F}^{+}_{i\fc j}\sheaf{F}_{j \fc  ij}(x_j[0]) & \\
        \succeq 
        \bigmeet_{j \in \graph{N}_{i} \cap \tau_{t}} & \sheaf{F}^{+}_{i\fc j}\sheaf{F}_{j \fc  ij}(x_j[0]) \\
        & \succeq x_i[0] ,
    \end{align*}
    using the definition of the asynchronous Tarski Laplacian. 
    Applying Lemma~\ref{lem:residuated} yields:
    \begin{align}
    \label{eq:succeq}
        \sheaf{F}_{j \fc  ij}(x_j[0]) & \succeq \sheaf{F}_{i \fc  ij}(x_i[0]).
    \end{align}
    Suppose $\eqref{eq:succeq}$ holds for a particular $i \in \graph{V}$ and a $j \in \graph{N}_{i} \cap \tau_{t}$.
    By liveness, there exist $t'\geq t$ such that $i \in \tau_{t'}$ so that $i \in \graph{N}_{j} \cap \tau_t$. In particular, there is a smallest $t'>t$ with
    \begin{align} 
    \label{eq:preceq}
        \sheaf{F}_{j \fc  ij}(x_j[t']) & \preceq \sheaf{F}_{i \fc  ij}(x_i[t']).
    \end{align}
    By hypothesis, $x_{i}[0] = x_i[t']$ for all $i \in \graph{V}$. Hence, equations \eqref{eq:succeq} and \eqref{eq:preceq} imply $\sheaf{F}_{i \fc  ij}(x_i[t]) =  \sheaf{F}_{j \fc ij}(x_j[t])$ for all $t$. Therefore, $\vec{x}[t]$ is a section.
\end{proof}
\textbf{}
The argument in the proof of Theorem \ref{thm:main} in immediately implies an iterative protocol we call \define{gossip} for an agent $i \in \graph{V}$ to converge to a harmonic state by firing a finite number of times. The protocol is exactly the localized heat flow (\ref{eq:heatflow}). 

Local assignments $x_i[0]\in \sheaf{F}(i)$ are initialized by all nodes, then, in a series of rounds $t = 0,1,2,\dots$, proceed as follows. Each node $i$ listens for adjacent nodes broadcasting their ``fused observation'' $\sheaf{F}_{j \fc ij}\left(x_j[t]\right),~j \in N_i$. Upon receiving the ``encoded'' message, $i$ applies the residual map $\sheaf{F}_{i \fc ij}^{+}$ to $\sheaf{F}_{j \fc ij}\left(x_j[t]\right)$. At the end of each round, $i$ aggregates all of the ``decoded'' messages, including the original local assignment $x_i[t]$, taking meets. The resulting element of $\sheaf{F}(i)$ is the local assignment for the new round, $x_i[t+1]$. Finiteness of $\sheaf{F}(i)$  implies the gossip algorithm terminates in finitely many iterations since $\vec{x}[t+1] \preceq \vec{x}[t]$, or else $\vec{x}[t]$ is a section.\footnote{In fact, a weaker condition, that $\prod_{i \in \graph{V}}$ satisfies a descending chain condition, is possible if we require $\sheaf{F}_{i \fc ij}$ to preserve joins of an arbitrary subset.}
% As for complexity, suppose each $i$ has oracle access to binary meets, structure-preserving maps, and their residuals.  Let $h \left( \lattice{Q} \right)$ denote the length of a maximal chain in $\lattice{Q}$; let $N$ denote the number of agents in the system; and let $H = \max_{i} h\left( \sheaf{F}(i) \right)$. For $d(i)$ denote the degree of vertex $i$ and $D = \max_i d(i)$, one clearly has convergence of any initial assignment to a section in time $O(NDH)$.

%%%%%%%%%%%%%%%%%%%%%%%%%%%%%%%%%%%%%%%%%%%%%%%%%%%%%%%%%%%%%%%%%%%%%%%%%%%%%%%%%%%%%%%%%
\section{Semantics}
\label{sec:semantics}
%%%%%%%%%%%%%%%%%%%%%%%%%%%%%%%%%%%%%%%%%%%%%%%%%%%%%%%%%%%%%%%%%%%%%%%%%%%%%%%%%%%%%%%%%

Examples of the Tarski Laplacian and heat flow are especially well-suited to distributed multimodal logic, where: vertices of a network correspond to agents; edges are communications between agents; assignments are sets of states associated to each vertex; and sections are assignments that express a consensus of knowledge across the network. For this we require the basic theory of Kripke semantics \cite{fagin2004reasoning}.

%%%%%%%%%%%%%%%%%%%%%%%%%%%%%%%%%%%%%%%%%%%%%%%%%%%%%%%%%%%%%%%%%%%%%%%%%%%%%%%%%%%%%%%%%
\subsection{Kripke semantics}
\label{sec:kripke}
%%%%%%%%%%%%%%%%%%%%%%%%%%%%%%%%%%%%%%%%%%%%%%%%%%%%%%%%%%%%%%%%%%%%%%%%%%%%%%%%%%%%%%%%%

\begin{definition}
A \define{Kripke frame} $\kripkeframe = \left(S,\relation{K}_1, \relation{K}_2, \dots, \relation{K}_n \right)$ consists of a set of \define{states}, $S$, and a sequence of binary relations $\relation{K}_i \subseteq S \times S$ used to encode modal operators. A \define{Kripke model} over a set $\Phi$ of atomic propositions consists of the data $\mathcal{M} = \left(\kripkeframe, \pi \right)$, where $\pi: S \to 2^{\Phi}$ validates whether or not a state $s \in S$ satisfies an atomic proposition $p \in \Phi$.
\end{definition}

Suppose $\phi$ is a formula; then one writes $(\mathcal{M}, s) \models \phi$ if $s$ {\em satisfies} the formula $\phi$ in the model $\mathcal{M}$. 
The semantics of a Kripke model is defined inductively using the symbols $\mathtt{true}$, $\wedge$, and $\neg$ in their typical usage:
\begin{enumerate}
\item $(\kripke{M},s) \models \mathtt{true}$ for all $s \in S$.
\item For $p \in \Phi$ atomic, $(\kripke{M},s) \models p$ if and only if $p \in \pi(s)$.
\item For $\phi$ an formula, $(\kripke{M}, s) \models \neg \phi$ if and only if $(\kripke{M}, s) \not \models \phi$.
\item For $\phi, \psi$ formulae, $(\kripke{M}, s) \models \phi \wedge \psi$ if and only if $(\kripke{M}, s) \models \phi$ and $(\kripke{M}, s) \models \psi$.
\end{enumerate}
There are additional (dual) modal operators on formulae, $K_i$ and $P_i$ (typically corresponding to knowledge and possibility), based on the binary relations $\relation{K}_i$. One writes $(\kripke{M}, s) \models {K}_i \phi$ if and only if $(\kripke{M}, t) \models \phi$ for all $t \in S$ such that $(s,t) \in \relation{K}_i$.
The dual operators $P_i$ are defined via $P_i \phi = \neg K_i \neg \phi$. Standard operations in propositional logic such as $\phi \to \psi$ are derived in the usual way \cite{mendelson2009introduction}. By abuse of notation, we write $\kripke{M} \models \phi$ if $(\kripke{M},s) \models \phi$ for all $s \in S$.

Depending on a number of axioms placed on the modal operators $K_i$ and $P_i$, one has rich interpretations for $K_i$ and $P_i$. For instance, the Knowledge Axiom \cite{fagin2004reasoning}
\begin{align}\label{eq:knowledge-axiom}
    \kripke{M} \models & K_i \phi \to \phi
\end{align}
and the Introspection Axiom \cite{fagin2004reasoning}
\begin{align}
    \kripke{M} \models & K_i \phi \to K_i K_i \phi
\end{align}
together suggest the interpretation of $K_i \phi$ and $P_i \phi$: \textit{agent $i$ knows $\phi$} and \textit{agent $j$ considers $\phi$ possible}, respectively. On the other hand, if \eqref{eq:knowledge-axiom} not hold, but the Consistency Axiom \cite{fagin2004reasoning} does hold,

\begin{align}
    \kripke{M} \models & \neg K_i (\mathtt{false}),
\end{align}
 $K_i \phi$ and $P_i \phi$ could be interpreted as: \textit{agent $i$ believes $\phi$} and \textit{agent $i$ does not disbelieve (i.e., is undecided about) $\phi$}.\footnote{Belief is assumed to satisfy the law of excluded middle: believing something does not make it true.}

The set of all finite formulae inductively obtained by $\Phi$ is called the \define{language} denoted $\kripke{L}(\Phi)$ (omitting the $\Phi$ where understood). For $\phi \in \kripke{L}$, and model $\kripke{M}$, the \define{intent} of $\phi$ is the subset
\[\phi^{\kripke{M}} = \{s \in S~\vert~(\kripke{M},s) \models \phi \}.\]
Formulae $\phi, \psi \in \kripke{L}$ are \define{semantically equivalent} if $\phi^\kripke{M} = \psi^\kripke{M}$. Semantic equivalence is an equivalence relation, written $\phi \equiv \psi$, with $\kripke{L}_{\equiv}(\kripke{M})$ denoting the set of equivalence classes in $\kripke{L}$ up to semantic equivalence.

%--------------------------------------------
\subsection{Semantic diffusion}
\label{sec:semanticdiffusion}
%--------------------------------------------
\noindent Our goal is to adapt the technology of sheaves of lattices and Laplacians to Kripke semantics over a network of agents. One simple approach is to use powerset lattices $2^S$ of states $S$ of a frame. Let $\graph{G}=(\graph{V},\graph{E})$ be a network and $\kripkeframe$ a frame. Define the \define{semantic sheaf} $\sheaf{I}$ over $\graph{G}$ so that the data over each vertex and edge is precisely $2^S$. The following is crucial to define the structure-preserving maps:
\begin{definition}\label{thm:exists-forall}
For $\kripkeframe = \left(S, \relation{K}_1, \dots, \relation{K}_n \right)$ a frame and for each $i \in \{0,1,\dots,n\}$, there is residual pair
\begin{equation}
\begin{tikzcd}
\bool^S \arrow[r,"\relation{K}_i^{\exists}",bend left=30] \arrow[r,"\bot", phantom] & \arrow[l, "\relation{K}_i^{\forall}", bend left = 30] \bool^S
\end{tikzcd}
\end{equation}
given by the formulae
\begin{align*}
\relation{K}_i^{\exists}(\sigma) &= \{t \in S~\vert~\exists  s \in \sigma,~(s,t) \in \relation{K}_i \}, \\
\relation{K}_i^{\forall}(\sigma) &= \{s \in S~\vert~\forall t \in S, (s,t) \in \relation{K}_i \Rightarrow t \in \sigma  \}.
\end{align*}
\end{definition}

\begin{lemma}
      Each map $\relation{K}_i^{\exists}: 2^S \to 2^S$ is join-preserving and $\left(\relation{K}_i^{\exists}\right)^{+}=\relation{K}_i^{\forall}$.
\end{lemma}
\begin{proof}
    First,
    \begin{align*}
        \relation{K}_i^{\exists}(\sigma \cup \sigma') &= \{t \in S~\vert~\exists~s~\text{or}~s' \in \sigma,~(s,t)~\text{or}~(s',t) \in \relation{K}_i \}.
    \end{align*}
    Second,
    \begin{align*}
        \left(\relation{K}_i^{\exists}\right)^{+}(\sigma) & = \bigcup \{\alpha \in 2^S~\vert~\relation{K}_i^{\exists}(\alpha) \subseteq \sigma \} \\
             & = \{s \in S~\vert~\forall t \in S, (s,t) \in \relation{K}_i \Rightarrow t \in \sigma \} \\
             & = \relation{K}_i^{\forall}(\sigma) .
    \end{align*}
\end{proof}

Such a sheaf of powerset lattices has an asynchronous Tarski Laplacian and a corresponding heat flow. The following definitions are straight translations from \S\ref{sec:async}.

\begin{definition}\label{def:semantic-laplacian}
    Suppose $\kripkeframe = (S, \relation{K}_1, \dots, \relation{K}_n)$ is a frame and $\sheaf{I}$ a sheaf of powersets $2^S$ over a network $\graph{G} = (\graph{V},\graph{E})$ where $\graph{V} = \{1,2, \dots, n\}$. Let $\tau: \N \to \graph{V}$ be a firing sequence. The (asynchronous) \textit{semantic Laplacian} is the operator acting on $\boldsymbol \sigma \in \prod_{i \in \graph{V}} 2^S$ via: 
    \begin{equation}
    \label{eq:semlap}
    \left( L_t \boldsymbol\sigma \right)_i 
    = 
    \bigcap_{j \in \graph{N}_{i} \cap \tau_t} 
    \relation{K}_i^{\forall}\relation{K}_j^{\exists}(\sigma_j) .
    \end{equation}
The associated \define{heat flow} is the dynamical system
    \begin{equation}
    \label{eq:positive-flow}
    {\boldsymbol \sigma}[t+1] = (\id \meet L_t) {\boldsymbol \sigma}[t]
    \end{equation}
\end{definition}

\noindent Our main result follows directly from Theorem \ref{thm:main}, interpreted in the language of this section.
\begin{theorem}
    Suppose $\graph{G} = (\graph{V}, \graph{E})$ is a network with firing sequence $\tau$ satisfying liveness. Let $\kripke{M} = (S, \relation{K}_1, \dots, \relation{K}_N, \Phi, \pi)$ be a Kripke model. Then, the sections of $\sheaf{I}$ are exactly time-independent solutions to the heat flow \eqref{eq:heatflow}.
\end{theorem}

These sections are interpretable as {\em possibility consensus} assignments of the model $\kripke{M}$ on $\graph{G}$. The assignment of formulae $(\phi_i)_{i\in\graph{V}}$ satisfies, for each edge $ij\in\graph{E}$, $P_i \phi_i \equiv P_j \phi_j$. This does not mean that the formulae are in consensus as identical formulae; rather, they are semantically equivalent in the language $\kripke{L}$.

%------------------------------------
\subsection{Syntactic diffusion}
\label{sec:syntacticdiffusion}
%------------------------------------

\noindent The following lemma motivates why local residuations $\relation{K}_i^\exists \dashv \relation{K}_i^\forall$ induce a flow of knowledge.

\begin{lemma}\label{thm:galois-semantics}
Suppose $\mathcal{M} = (\kripkeframe, \pi)$ is a model and $\phi \in \kripke{L}(\kripke{M})$. Then, the following hold
\begin{align}
\relation{K}_i^{\forall}\left( \phi^{\kripke{M}}\right) &= \left( {K}_i \phi \right)^{\kripke{M}}, \label{eq:galois-semantics-1}\\
\relation{K}_i^{\exists}\left( \phi^{\kripke{M}} \right) &= \left( P_i \phi \right)^{\kripke{M}}. \label{eq:galois-semantics-2}
\end{align}
\end{lemma}
\begin{proof}
    Writing
    \[\relation{K}_i^{\forall}\left( \phi^{\kripke{M}}\right) = 
         \{s \in S~\vert~\forall t \in S, (s,t) \in \relation{K}_i \Rightarrow (\kripke{M},t) \models \phi \},\]
    we prove \eqref{eq:galois-semantics-1}. For \eqref{eq:galois-semantics-2}, evaluating $\relation{K}_i^{\exists}(\phi^\kripke{M})$ yields
    \begin{align*}
        \{t \in S~\vert~\exists  s~\text{such that}~(\kripke{M},s) \models \phi~\text{with}~(s,t) \in \relation{K}_i \}.
    \end{align*}
    On the other hand, evaluating $\left(P_i \phi \right)^\kripke{M}$ yields
    \begin{align*}
        \left( \neg K_i \neg \phi \right)^{\kripke{M}} & & \\
        = & \left( \left( K_i \neg \phi \right)^{\kripke{M}} \right)^c & \\
        = & \{s~\vert~ \forall~t~\text{such that}~(s,t) \in \relation{K}_i,~(\kripke{M},t) \not \models \phi \}^{c} & .
    \end{align*}
\end{proof}

\noindent The following simple observation together with Lemma \ref{thm:galois-semantics} allows us to freely go back and forth between syntax and semantics, opening the way for syntactic diffusion dynamics.

\begin{lemma}\label{lem:contravariant}
Suppose $\{\phi\}_{i}$ is a finite set of formula in $\kripke{L}(\kripke{M})$. Then, $\left(\bigjoin_{i \in I} \phi_i\right)^{\kripke{M}} = \bigcup_{i} \phi_i^{\kripke{M}}$ and $\left(\bigmeet_{i} \phi_i\right)^{\kripke{M}} = \bigcap_{i} \phi_i^{\kripke{M}}$.
\end{lemma}

\begin{definition}
\label{def:syntactic-laplacian}
Suppose $\boldsymbol \phi = (\phi_i)_{i \in \graph{V}}$ is a tuple of formulae in $\kripke{M}$.
% and $\boldsymbol \phi^{\kripke{M}} = \prod_{i \in \graph{V}} \phi_i^{\kripke{M}}.$
The (asynchronous) \define{syntactic Laplacian} acts on assignments as:
% \begin{equation} 
%         \left( L_t \{ \boldsymbol\phi^{\kripke{M}} \}\right)_i = \left( \bigmeet_{j \in \graph{N}_{i} \cap \tau_t} K_i P_j \phi_j \right)^{\kripke{M}} .
% \end{equation}
%THINK YOU ACTUALLY WANT THIS...
\begin{equation}\label{synlap}
        \left( L_t \boldsymbol\phi\right)_i = \bigmeet_{j \in \graph{N}_{i} \cap \tau_t} K_i P_j \phi_j  .
\end{equation}
\end{definition}
This is a straight translation of the semantic Tarski Laplacian, using Lemmas \ref{thm:galois-semantics} and \ref{lem:contravariant}.

%------------------------------------
\subsection{Dual Laplacians}
\label{sec:dual}
%------------------------------------

\noindent The classical meet-join duality in lattices pushes through to all other structures built therefrom. In particular, the Tarski (and thus semantic and syntactic) Laplacians come in dual variants, implicating how syntactic and semantic consensus is interpreted. 

\begin{definition}
\label{def:dual}
The \define{dual} Tarski, semantic, and syntactic Laplacians are given by, respectively:
\begin{equation}
  \label{eq:dual-tarski}
  \left(L^*\vec{x}\right)_i 
  = 
  \bigvee_{j \in \graph{N}_{i}} \sheaf{F}^{+}_{i\fc j} \sheaf{F}_{j \fc  ij}(x_j) 
\end{equation}
\begin{equation}
  \label{eq:dual-semantic}
  \left(L^* \boldsymbol\sigma\right)_i 
  = 
  \bigcup_{j \in \graph{N}_{i}} {\relation{K}}_{i}^{\exists}{\relation{K}}_j^{\forall}(\sigma_j)
\end{equation}

\begin{equation}
  \label{eq:dual-syntactic}
  \left(L^* \boldsymbol\phi \right)_i 
  = 
  \bigjoin_{j \in \graph{N}_{i}} P_i K_j \phi_j.  \\
\end{equation}
The corresponding asynchronous dual Laplacians $L^*_t$ intersect neighboring vertices with those in a firing sequence $\tau$.
\end{definition}

The \define{dual heat flow} iterates the operator $\id\join L^*_t$. Convergence results to sections remain. These sections are interpretable as {\em knowledge consensus} assignments of the model $\kripke{M}$ on $\graph{G}$. The assignment of formulae $(\phi_i)_{i\in\graph{V}}$ satisfies, for each edge $ij\in\graph{E}$, $K_i \phi_i \equiv K_j \phi_j$.

%%%%%%%%%%%%%%%%%%%%%%%%%%%%%%%%%%%%%%%%%%%%%%%%%%%%%%%%%%%%%%%%%%%%%%%%%%%%%%%%%%%%%
\section{Examples}
\label{sec:examples}
%%%%%%%%%%%%%%%%%%%%%%%%%%%%%%%%%%%%%%%%%%%%%%%%%%%%%%%%%%%%%%%%%%%%%%%%%%%%%%%%%%%%%

In this section, we supply an example of a semantic sheaf modeling knowledge gain; we also provide a numerical experiment demonstrating the correctness of asynchronous heat flow \eqref{eq:heatflow} as well as convergence behavior. The \define{Lyapunov function} acts on assignments in \eqref{eq:heatflow} via
\begin{equation}
V(\vec{x}) = \sum_{ij \in \graph{E}} d\left(\sheaf{F}_{i \fc ij}(x_i),\sheaf{F}_{j \fc ij}(x_j)\right); \label{eq:lyp}
\end{equation}
each $d$ is a distance function on $\sheaf{F}(ij)$. If $d$ is a metric, then $V(\vec x) \geq 0$ and $V(\vec x) = 0$ if and only if $\vec x$ is a section.
% Our notion of a Lyapunov funtion on a sheaf is reminiscent of the already-studied \define{consistency radius} \cite{robinson2020assignments}.
Below, the metric is taken to be the well-known \define{Jaccard distance} between subsets.
% \[
% d_J(\sigma, \sigma') = \frac{\vert \sigma \cup \sigma' \vert - \vert \sigma \cap \sigma' \vert}{\vert \sigma \cup \sigma' \vert}.
% \]

%------------------------------------
\subsection{Threat detection}
\label{sec:targets}
%------------------------------------

% \noindent In this example, we model information gain by four sensors ($N=4$) with two-way communication links $\graph{E} = \{(1,2),(2,3),(3,4),(1,4)\}$. All four sensors are tasked with detecting malicious targets $A$ and $B$. Each sensor $i$ has a local status $s_i \in S_i=\{\hat{0},\hat{1}\}$, so that $S = \prod_{i=1}^{i=4} S_i$. The atomic propositions are $\Phi = \{p_A, p_B\}$ which are interpreted as {\em threat A is present}, or {\em threat B is present}. A global state $s = (s_1,s_2,s_3,s_4)$ may be thought of as a signature for the presence or absence of a threat via any given $\pi: S \to 2^{\{p_A, p_B\}}$.
% In the context of knowledge, a natural choice of Kripke relation $\relation{K}_i$ is the equivalence relation on $S$ given by $s \sim_i t$ if and only if $s_i = t_i$ \cite{fagin2004reasoning}. This reflects the intuition that $i$ can distinguish system-wide states $s$ and $t$ insomuch as $i$ can distinguish local statuses with its own instruments.

\noindent In this example, we model the knowledge of $N$ ``smart'' sensors tasked with detecting $M$ possible targets $\Phi = \{A_1, A_2, \dots, A_M\}$. Sensors are equipped with two-way links forming a communication pattern modeled by an undirected graph $\graph{G} = (\graph{V}, \graph{E})$. Each sensor $i$ has a local state $s_i$ and a set of possible local states $S_i$ which may vary from sensor to sensor. For instance, $s_i$ could represent a risk posed at a particular location. The atomic propositions $\Phi$ are each interpreted as \emph{threat $A$ is present}. A global state $s \in \prod_{i=1}^N S_i$ is a tuple of all local states. The ground truth on whether or not threat $A$ is present is determined by a map $\pi: \prod_{i=1}^N S_i \to 2^{\Phi}$. To model the knowledge of individual sensors, the right choice of Kripke relation $\relation{K}_i$ is the equivalence relation on $S = \prod_{i=1}^N S_i$ given by $s \sim_i t$ if and only if $s_i = t_i$, reflecting the intuition that a sensor $i$ would know $A$ is present -- or, in general, any formula in $\mathcal{L}(\Phi)$ -- precisely when, given a current local state $s_i \in S_i$, $A$ was present for every possible global state that had $s_i$ as a local state.

\subsection{Numerical experiments}
\label{sec:experiments}
%------------------------------------

\noindent In the following simulation, we demonstrate the validity of the gossip algorithm by: 1) generating a random geometric graph ($N=40, r=0.08$); 2) assigning (Kripke) relation on $S$ ( with$10$ states), one for each node, by wiring each $(x,y) \in \relation{K}_i$ randomly ($p = 0.9$ if $x = y$, $p=0.1$ otherwise); 3) selecting several random firing sequence $\tau: \{0,1,\dots \} \to 2^\graph{V}$ which determine which nodes broadcast to their neighbors at time $t$. For a random initial assignment $\boldsymbol \sigma [0]$, we run the heat flow dynamics using the dual asynchronous semantic Laplacian (\ref{eq:dual-semantic}) for each firing sequence (Fig.~\ref{fig:1}).

\begin{figure}[h]
\includegraphics[width=0.5\textwidth]{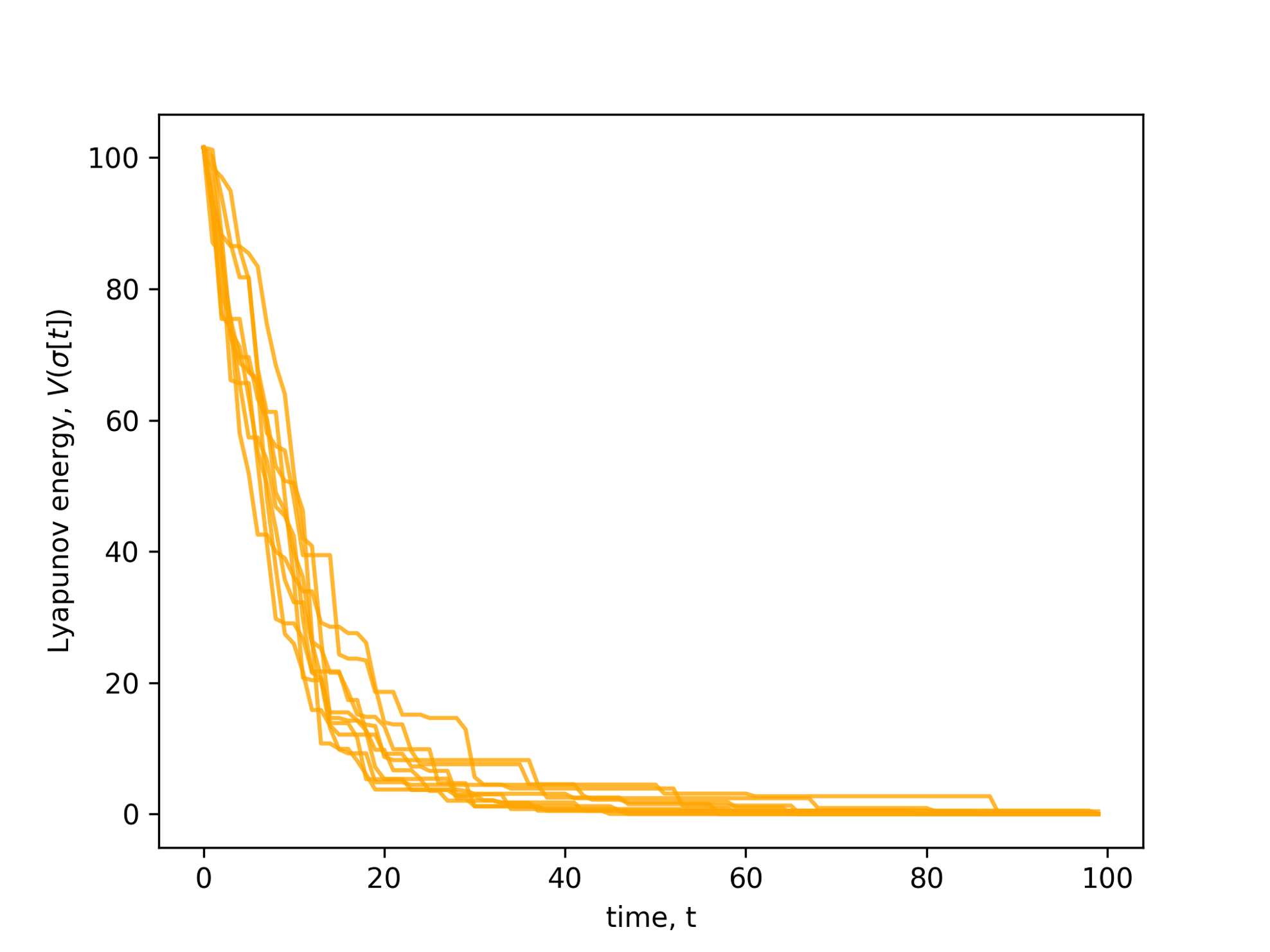}
\caption{The Lyapunov energies of each iteration of the gossip algorithm computed from the dual asynchronous semantic Laplacian of the Kripke frame defined above. Each trial is a different firing sequence $\tau$.}
\label{fig:1}
\end{figure}

%%%%%%%%%%%%%%%%%%%%%%%%%%%%%%%%%%%%%%%%%%%%%%%%%%%%%%%%%%%%%%%%%%%%%%%%%%%%%%%%%%%%%
%%%%%%%%%%%%%%%%%%%%%%%%%%%----BIBLIOGRAPHY----%%%%%%%%%%%%%%%%%%%%%%%%%%%%%%%%%%%%%%
%%%%%%%%%%%%%%%%%%%%%%%%%%%%%%%%%%%%%%%%%%%%%%%%%%%%%%%%%%%%%%%%%%%%%%%%%%%%%%%%%%%%%

\bibliographystyle{ieeetr}
\bibliography{bibliography}

\end{document}